\documentclass[submission,copyright,creativecommons]{eptcs}
\providecommand{\event}{LSFA 2022}
\usepackage{underscore}    
\usepackage{ucs}
\usepackage[utf8]{inputenc}
\usepackage{amsfonts}
\usepackage{amsmath}
\usepackage{amsthm}
\usepackage{amssymb}
\usepackage{booktabs}
\usepackage[capitalize]{cleveref}
\usepackage{fancyvrb}
\usepackage{newunicodechar}
\usepackage{tipa} 
\usepackage[inline]{enumitem}
\usepackage{comment}

\DefineVerbatimEnvironment
  {code}{Verbatim}
  {commandchars=\\`',codes={\catcode`$=3}} 

\DefineShortVerb[commandchars=\\`^,codes={\catcode`$=3}]{\|}

\newunicodechar{⁻}{${}^-$}
\newunicodechar{≢}{$\not\equiv$}
\newunicodechar{ℕ}{$\mathbb{N}$}
\newunicodechar{∶}{:}
\newunicodechar{⊢}{$\vdash$}
\newunicodechar{′}{'}
\newunicodechar{≡}{$\equiv$}
\newunicodechar{λ}{$\lambda$}
\newunicodechar{ν}{$\nu$}
\newunicodechar{∷}{::}
\newunicodechar{Γ}{$\Gamma$}
\newunicodechar{⊔}{$\sqcup$}
\newunicodechar{⊥}{$\bot$}
\newunicodechar{₁}{${}_1$}
\newunicodechar{∘}{$\circ$}
\newunicodechar{∈}{$\in$}
\newunicodechar{Σ}{$\Sigma$}
\newunicodechar{▹}{$\triangleright$}
\newunicodechar{ℓ}{$\ell$}
\newunicodechar{∼}{$\sim$}
\newunicodechar{⇂}{$\downharpoonright$}
\newunicodechar{ι}{$\iota$}
\newunicodechar{χ}{$\chi$}
\newunicodechar{∙}{$\bullet$}
\newunicodechar{≟}{$\stackrel{?}{=}$}
\newunicodechar{≺}{$\prec$}
\newunicodechar{σ}{$\sigma$}
\newunicodechar{τ}{$\tau$}
\newunicodechar{β}{$\beta$}
\newunicodechar{⇒}{$\Rightarrow$}
\newunicodechar{Λ}{$\Lambda$}
\newunicodechar{ƛ}{\textipa{\textcrlambda}}
\newunicodechar{∀}{$\forall$}
\newunicodechar{α}{$\alpha$}
\newunicodechar{⟿}{$\rightsquigarrow$}

\newcommand{\unary}[2]{{\tt [\,\ensuremath{#1}\,/\,\ensuremath{#2}\,]}}
\newcommand{\effect}[2]{{\tt \ensuremath{#1}\,$\bullet$\,\ensuremath{#2}}}
\newcommand{\upd}[3]{{\tt \ensuremath{#1}\,$\prec$+\,(\ensuremath{#2}\,,\,\ensuremath{#3})}}

\newcommand{\sn}[1]{{\tt sn}\ensuremath{\,#1}}
\newcommand{\sng}{{\tt sn}}
\newcommand{\net}[1]{{\tt Ne}\ensuremath{\,#1}}

\newcommand{\Red}[2]{{\tt Red}\ensuremath{\,#1\,#2}}
\newcommand{\RedSubst}[2]{{\tt RedSubst}\ensuremath{\,#1\,#2}}

\newcommand{\Redg}{{\tt Red}}

\newcommand{\fresh}[2]{{\tt \ensuremath{#1}\,\#\,\ensuremath{#2}}}
\newcommand{\free}[2]{{\tt \ensuremath{#1}\,*\,\ensuremath{#2}}}
\newcommand{\restr}[2]{{\tt (\ensuremath{#1}\,,\,\ensuremath{#2})}}
\newcommand{\freshr}[3]{{\tt \ensuremath{#1}\,$\#\downharpoonright$\,(\ensuremath{#2}\,,\,\ensuremath{#3)}}}
\newcommand{\recg}{{\tt Rec}}
\newcommand{\rec}[3]{{\tt Rec}\ensuremath{\,#1\,#2\,#3}}
\newcommand{\zero}{{\tt O}}
\newcommand{\sucg}{{\tt S}}
\newcommand{\suc}[1]{{\tt S}\,\ensuremath{#1}}
\newcommand{\stepg}{{\tt →}\ensuremath{\beta}}
\newcommand{\step}[2]{\ensuremath{#1} {\tt →}\ensuremath{\beta} \ensuremath{#2}}
\newcommand{\stepc}[2]{\ensuremath{#1 \rightsquigarrow #2}}
\newcommand{\stepcg}{\ensuremath{\rightsquigarrow}}
\newcommand{\conv}[2]{\ensuremath{#1\ {\sim}\alpha\ #2}}
\newcommand{\convg}{\ensuremath{{\sim}\alpha}}

\newcommand{\contracts}[2]{\ensuremath{#1 \triangleright #2}}
\newcommand{\contraction}{\ensuremath{\triangleright}}
\newcommand{\betacontraction}{\ensuremath{\triangleright\beta}}

\newcommand{\anonymdef}[1]{\underline{#1}}

\theoremstyle{definition}
\newtheorem{theorem}{Theorem}[section]
\newtheorem{definition}[theorem]{Definition}
\newtheorem{lemma}[theorem]{Lemma}

\title{A Formal Proof of the Strong~Normalization~Theorem for System T in Agda\thanks{This work is partially supported by Agencia Nacional de Investigaci\'on e Innovaci\'on (ANII), Uruguay.}}

\author{Sebasti\'an Urciuoli
\institute{Universidad ORT Uruguay\\ Montevideo, Uruguay}
\email{urciuoli@ort.edu.uy}
}
\def\titlerunning{A Formal Proof of the Strong Normalization Theorem for System T}
\def\authorrunning{S. Urciuoli}

\begin{document}

\maketitle

\begin{abstract}
We present a framework for the formal meta-theory of lambda calculi in first-order syntax, with two sorts of names, one to represent both free and bound variables, and the other for constants, and by using Stoughton's multiple substitutions.
On top of the framework we formalize Girard's proof of the Strong Normalization Theorem for both the simply-typed lambda calculus and System~T.
As to the latter, we also present a simplification of the original proof. The whole development has been machine-checked using the Agda system.
\end{abstract}

\section{Introduction}

In \cite{tasistro2015} a framework was presented for the formal meta-theory of the pure untyped lambda calculus in first-order abstract syntax (FOAS) and using only one sort of names for both free and bound variables\footnote{Both the previous framework and the one presented  here use named variables, it bears repeating. In a contrary sense, there are nameless approaches, e.g., de-Bruijn indices \cite{debruijn1972} or locally nameless syntax \cite{locallynameless}, which use numbers to identify the variables.}.
Based upon Stoughton's work on multiple substitutions \cite{stoughton88}, the authors were able to give a primitive recursive definition of the operation of substitution which does not identify alpha-convertible terms\footnote{Or without using Barendregt's variable convention.},
avoids variable capture, and has a homogeneous treatment in the case of abstractions.
Such a definition of substitution is obtained by renaming every bound name to a sufficiently fresh one. 
The whole development has been formalized in constructive type theory using the Agda system \cite{agda}.

The framework has been used 
since then to verify many fundamental meta-theoretic properties of the lambda calculus including:
Subject Reduction for the simply-typed lambda calculus (STLC) in \cite{copello2016};
the Church-Rosser Theorem for the untyped lambda calculus also in \cite{copello2016};
the Standardization Theorem in \cite{copes2018}, and;
the Strong Normalization Theorem for STLC in \cite{urciuoli2020}, and by using F.~Joachimski and R.~Matthes' syntactical method \cite{joachimski2003}.
Now in this paper, we continue the same line of work and formalize the Strong Normalization Theorem for System T, and we also present a new and different mechanization for STLC.

System~T extends STLC by adding primitive recursive functions on natural numbers.
It has its roots in K.~Gödel's work presented in \cite{godel1958}, and it was originally developed to study the consistency of Peano Arithmetic.
The Strong Normalization Theorem states 
that every program (term) in some calculus under consideration is strongly normalizing.
A term is \textit{strongly normalizing} if and only if its computation 
always halts regardless of the reduction path been taken. 
This result for System~T is already well known.
In this development we mechanize J.-Y.~Girard's proof presented in \cite{girard1989}, which in turn is based on W.~W.~Tait's method of {\em computability} or {\em reducible functions} \cite{tait} (henceforth we shall refer to Girard and Tait's method or proof interchangeably). This method defines a (logical) relation between terms and types that is fitter than the Strong Normalization Theorem, and hence it enables a more powerful induction hypothesis. 
Any term related to some type under such a relation is said to be \textit{reducible}.
Then the method consists of two steps: first, to prove that all reducible term are strongly normalizing, and secondly to prove that all typed terms are reducible. 

Initially, the sole objective of this work was to formalize a proof of the Strong Normalization Theorem but only for System~T, and by using the framework presented in \cite{tasistro2015}. 
Of course, the syntax of the pure lambda terms had to be extended to include the term-formers for the natural numbers and the recursion operator\footnote{In contrast to \cite{girard1989}, during this development we shall not consider booleans nor tuples as part of the syntax. Nevertheless, they can be easily defined by the machinery presented here.}. 
For this, we based ourselves upon a standard definition of the lambda terms in which two disjoint sort of names are used, one to represent the variables, and the other for the constants, e.g., see \cite{hindley}.
Now, instead of restricting ourselves to a specific set of constants, we shall allow any (countable) set. 
Once the syntax of the framework had been parameterised 
it felt natural to parameterise the reduction schema as well, as these relations are often defined by the syntax.
The work went a bit further, and the first part of the proof was also abstracted for a class of calculi to be defined; this step consists mainly in analysing reduction paths.
To round up, hitherto the work evolved from formalizing the proof of the Strong Normalization Theorem in System~T, into also providing a general-purpose framework with theories for substitution, alpha-conversion, reduction and reducible terms of simple types. 

Now, having such a framework it was a good time to revisit the previous formalization of the Strong Normalization Theorem for STLC presented in \cite{urciuoli2020}. 
There, the definition of the logical relation was based on the one in the POPLMark Challenge 2 \cite{poplmark2}, and it included the context of variables.
In addition to that, a syntactical characterization based on \cite{joachimski2003} was used to define the type of the strongly normalizing terms.
In this development, we shall use a standard definition of the logical relation which does not contain the context, and an accessibility characterization of the strongly normalizing terms based on \cite{alti:phd93}. 
Furthermore, the proof for STLC is contained in the one for System~T, so it serves both as a milestone in this exposition, as well as to show the incremental nature of the whole method presented here.

The last result presented in this development is about a simplification in Girard's proof of the Strong Normalization Theorem for System~T.
More specifically, in the second part of the proof there is a lemma whose principle of induction requires to count the occurrences of the successor operator in the \textit{normal form} of a given strongly normalizing term. 
This is not strictly necessary, and one can just count such symbols \textit{directly} in the term, and so \textit{avoid evaluating} it. 

In summary, the novel contributions in this paper are:
(1) a framework for the meta-theory of lambda calculi in FOAS with named variables and constants;
(2) a complete mechanization of Girard's proof of The Strong Normalization Theorem for System~T in Agda;
(3) a new and different mechanization of Girard's proof for STLC in Agda as well, and;
(4) a simplification of the principle of induction in Girard's original proof of The Strong Normalization Theorem for System T.
To the best of our knowledge, there is not yet a mechanization of the Strong Normalization Theorem for System~T.
The development has been entirely written in Agda and it is available at: \url{https://github.com/surciuoli/lambda-c}.

The structure of this paper is the following. In the next section we introduce the new framework: its syntax, substitution, conversion theories and logical relations (reducible terms). 
Some results presented are completely new, and some others are an extension of \cite{tasistro2015,urciuoli2020} to consider the additional syntax. 
From \cref{subsec:reducibility} on, and unless the opposite is explicitly stated, all results represent new developments. 
In \cref{sec:stlc}, we formalize both STLC and Girard's proof of the Strong Normalization Theorem. In \cref{sec:systemt}, we extend both the calculus and Girard's proof to System~T, and we also explain the aforementioned simplification. 
In the last sections we give some overall conclusions and compare our work with related developments. 

Throughout this exposition we shall use Agda code for definitions and lemmata, and a mix of code and English for the proofs in the hope of making reading more enjoyable. A certain degree of familiarity with Agda or at least with functional programming languages like Haskell is assumed. 

\section{The Framework}
\label{sec:framework}

Let $\mbox{{\tt V}} = v_0, v_1 \ldots$ be any infinitely countable collection of names, the \anonymdef{variables}, ranged over by letters $x$, $y\ldots$
and equipped with a deciding procedure for definitional equality; 
for concreteness, we shall define ${\tt V} = \mathbb{N}$, i.e., the set of natural numbers in Agda,
but it can be any other suitable type, e.g., strings.
Let |C| be any possibly infinite countable collection of names, the \anonymdef{constants}, and ranged over by $c$. The \anonymdef{abstract syntax} of the lambda terms with constants is defined:

\begin{definition}[Syntax]
\label{def:syntax}
\hfill 
\begin{code}[numbers=left]
module CFramework.CTerm (C : Set) where
...
data Λ : Set where
  k : C → Λ 
  v : V → Λ
  ƛ : V → Λ → Λ
  _·_ : Λ → Λ → Λ
\end{code}
\end{definition}
\noindent In line 1 we indicate that the definition is contained in the module |CFramework.CTerm|, which according to Agda's specification must be located in the file CFramework/CTerm.agda. 
We also specify that the module is parameterised by the set of constants |C|, which can be of any inductive type (|Set|). 
Lines 4 and 5 define the constructors for the constants and the variables respectively.
In line 6 we use {\tt ƛ} to not interfere with Agda's primitive $\lambda$.
We shall follow the next convention unless the opposite is explicitly stated: use $\lambda$ to represent \textit{object-level} abstractions in informal discussions and proofs, and use {\tt ƛ} in code listings. 
Line 7 defines the infix binary operator of function application.
As usual, we shall use letters $M$, $N\dots$ to range over terms.

The module can be then instantiated with any type of constants. 
For example, the next declaration derives the syntax of the \anonymdef{pure lambda terms} into the current scope:

\begin{definition}
\label{def:puresyntax}
|open import CFramework.CTerm ⊥|
\end{definition}

\noindent \texttt{⊥} is the inductive type without any constructor. 
The {\tt import} statement tells Agda to load the content of the file named after the module into the current scope, while the {\tt open} statement lets one access the definitions in it without having to qualify them.
Both statements can be combined into a single one as shown.

Whenever a name $x$ syntactically occurs in a term $M$ and is not bound by any abstraction, we shall say $x$ is \anonymdef{free} in $M$, and write it 
\free{x}{M}. On the other hand, if every occurrence of $x$ is bound by some abstraction (or even if $x$ does not occur at all), we shall say $x$ is \anonymdef{fresh} in $M$, and write it \fresh{x}{M} as in nominal techniques, e.g., see \cite{urban2004}. Both relations are inductively defined in a standard manner, and in \cite{tasistro2015} it was proven that both relations are opposite to each other.

It will come in handy to define both the type of predicates and binary relations on terms respectively by: |Pred = Λ → Set|, and |Rel = Λ → Λ → Set|.  

\subsection{Substitution}
\label{sec:substitution}
Substitution is the fundamental entity on which alpha- and beta-conversion sit. We shall base ourselves upon the work done in \cite{stoughton88}, and first define \anonymdef{multiple substitutions} as functions from variables to terms:
\begin{verbatim}
Subst = V → Λ
\end{verbatim}
\noindent We shall use letter $\sigma$ to range over them. 
Later, by applying these functions to the free variables in a given term we shall obtain the desired operation of the \textit{action of substitution} (\cref{def:actionSubst}), i.e., the operation of replacing every free name $x$ in $M$ by its corresponding image $\sigma x$.

Most substitutions appearing in properties and definitions are identity-almost-everywhere.
We can generate them by starting from
the \anonymdef{identity} substitution $\iota$, which maps every variable to itself, and applying the \anonymdef{update} operation on substitutions |_≺+_| such that for any $\sigma$, $x$ and $M$, \upd{\sigma}{x}{M} is the substitution that maps $x$ to $M$, and $y$ to $\sigma y$ for every $y$ other than $x$:
\begin{definition}[Update operation]
\label{def:update}
\hfill 
\begin{code}[numbers=left]
_≺+_ : Subst → V × Λ → Subst
(σ ≺+ (x , M)) y with x ≟ y
... | yes _ = M
... | no  _ = σ y
\end{code}
\end{definition}
\noindent In line 1, |×| is the non-dependent product type, and in line 2, |≟| is the procedure that decides if two names are equal, and mentioned at the start of this section.

In some places we shall need to restrict the domain of a substitution so to have a finite image or range, therefore we introduce the type of \anonymdef{restrictions}, written |R|, and defined: |R = Subst × Λ|.
Below we extend freshness to restrictions:
\begin{definition}[Freshness on restrictions]
\hfill 
\label{def:freshRestr}
\begin{verbatim}
_#⇂_ : V → R → Set
x #⇂ (σ , M) = (y : V) → y * M → x # σ y
\end{verbatim}
\end{definition}
\noindent In English, a name is fresh in the restriction \restr{\sigma}{M} if and only if it is fresh in every image $\sigma y$, for every \free{y}{M}.

Now we shall briefly discuss the mechanism in the framework used 
to rename the bound names in a given term, and so avoid capturing any free variable during the action of substitution. 
The complete description can be found in \cite{tasistro2015}.
Let |$\chi$'| be the function that returns the first name not in a given list:
\begin{verbatim}
χ' : List V → V    
\end{verbatim}
The algorithm is obtained by a direct consequence of the pigeonhole principle: the list of names given is finite, therefore we can always choose a fresh name from the infinite collection |V|. Then we can define the \anonymdef{choice} function |χ| that returns the first name not in a given restriction
\restr{\sigma}{M},
by first concatenating into a single list every free name that appears in the image $\sigma x$ for every \free{x}{M}, and then selecting the first name not in such a list by using the previous $\chi'$ function:
\begin{verbatim}
χ : R → V
χ (σ , M) = χ' (concat (mapL (fv ∘ σ) (fv M)))
\end{verbatim}
|mapL| applies a function to every element in a list, $\circ$ stands for the usual composition of functions, and |fv| computes the list of free names in a given term.
In \cite{tasistro2015} it was proven that $\chi$ computes a sufficiently fresh name, according to our expectations to be addressed shortly:
\begin{lemma}
\label{lemma:chi}
\begin{verbatim}
χ-lemma2 : (σ : Subst) (M : Λ) → χ (σ , M) #⇂ (σ , M)
\end{verbatim}
\end{lemma}

The \anonymdef{action of a substitution} $\sigma$ on a term $M$ is the operation that replaces every free name in $M$ by its corresponding image under $\sigma$. It is written \effect{M}{\sigma} and defined:
\begin{definition}[Action of substitution]
\label{def:actionSubst}
\hfill 
\begin{verbatim}
_∙_ : Λ → Subst → Λ
k c ∙ σ = k c
v x ∙ σ = σ x
M · N ∙ σ = (M ∙ σ) · (N ∙ σ)
ƛ x M ∙ σ = ƛ y (M ∙ σ ≺+ (x , v y)) where y = χ (σ , ƛ x M)
\end{verbatim}
\end{definition}
\noindent Notice that in the last equation we always rename the bound variable $x$ to $y$ by using the $\chi$ function. We can show that this method avoids variable capture:
for any \free{w}{M} other than $x$ it must follow \fresh{y}{(\mbox{\upd{\sigma}{x}{y}})w}, otherwise it would mean that we have captured an {\em undesired} free occurrence of $y$. 
Notice that if $w=x$ then its image is $y$ which represents an occurrence of $x$ in the original term $\lambda x M$ and therefore must be ``re-bound''.
So, \free{x}{M} and $x\not=w$, therefore \free{w}{\lambda x M}. Next,
by \cref{lemma:chi} we have \freshr{y}{\sigma}{\lambda x M}. Then, by \cref{def:freshRestr} it follows \fresh{y}{\sigma w}, and since $(\mbox{\upd{\sigma}{x}{y}})w=\sigma w$ by \cref{def:update}, so \fresh{y}{(\mbox{\upd{\sigma}{x}{y}})w}.

Unary substitution is defined: 
\begin{verbatim}
_[_/_] : Λ → Λ → V → Λ
M [ N / x ] = M ∙ ι ≺+ (x , N)
\end{verbatim}

Our definition of $\bullet$ has a direct consequence on the terms: when submitted to substitutions, the bound variables become ``ordered'', 
for the lack of a better name\footnote{There is a similitude between this designation and A. Church's definition of {\em principal normal form} terms in \cite[p.~348]{church36}.}. 
Consider the next example. 
Let $M = \lambda v_1 v_1$. 
By definition, $\mbox{\effect{M}{\sigma}} = \lambda x (v_1 \bullet \mbox{\upd{\sigma}{v_1}{x}}) = \lambda x x$, 
where $x = \chi\mbox{\restr{\sigma}{\lambda v_1 v_1}}$, and for every $\sigma$.
We can see that $M$ does not contain any free variable, therefore by definition of $\chi$ we have that $x = v_0$, i.e., the first name in |V|, and so we have that the closed term $\lambda v_1 v_1$ turned into $\lambda v_0 v_0$ even though no substitution actually happened. 
Another example a bit more sophisticated is the next one: 
$(\lambda v_3 \lambda v_2 \lambda v_0 (v_0 v_1 v_2 v_3)) \mbox{\unary{v_0}{v_1}} = \lambda v_1 \lambda v_2 \lambda v_3 (v_3 v_0 v_2 v_1)$.
This collateral effect will have some implications on our definition of beta-reduction. 

\subsection{Alpha-conversion}
\label{subsec:alpha}
Alpha-conversion is inductively defined by the syntax:
\begin{code}[numbers=left]
module CFramework.CAlpha (C : Set) where
open import CFramework.CTerm C
...
data _∼α_ : Rel where
  ∼k : {c : C} → k c ∼α k c
  ∼v : {x : V} → v x ∼α v x
  ∼· : {M M' N N' : Λ} → M ∼α M' → N ∼α N' → M · N ∼α M' · N'
  ∼ƛ : {M M' : Λ} {x x' y : V} → y # ƛ x M → y # ƛ x' M' 
     → M [ v y / x ] ∼α M' [ v y / x' ] → ƛ x M ∼α ƛ x' M'
\end{code}
\noindent Since the syntax of the lambda terms is parameterised by a set |C|, every module that depends on the syntax (all of them) will have to be parameterised by |C| as well. Lines 1 and 2 illustrate this point.

Arguments written between braces |{| and |}| are called \textit{implicit} and they are not required to be supplied; 
the type-checker will infer their values, whenever possible.
Implicit arguments can be made explicit by enclosing them between braces, e.g., |∼k {c$_1$}| has type |k c$_1$ ∼α k c$_1$|. 

The only case in the definition worth mentioning is |∼ƛ|. There, we rename both $x$ and $x'$ to a common fresh name $y$. If such results are alpha-convertible, then the choice of the bound name is irrelevant, and it should be expected to assert that both abstractions are alpha-convertible. This definition can also be seen in nominal techniques, e.g., see \cite{urban2004}, though there it happens to be more usual to rename only one side of \convg. Our symmetrical definition has some advantages over those that are not (see \cite{tasistro2015}).
Also, in \cite{tasistro2015}, \convg\ was proven to be an equivalence relation.

The next results are quickly extended from \cite{tasistro2015}: 

\begin{lemma}
\label{lemma:substIota}
\begin{verbatim}
lemma∙ι : ∀ {M} → M ∼α M ∙ ι
\end{verbatim}
\end{lemma}

\begin{lemma}
\label{lemma:composition}
\begin{verbatim}
corollary1SubstLemma : ∀ {x y σ M N} → y #⇂ (σ , ƛ x M) 
→ (M ∙ σ ≺+ (x , v y)) ∙ ι ≺+ (y , N) ∼α M ∙ σ ≺+ (x , N)
\end{verbatim}
\end{lemma}
\noindent Arguments preceded by $\forall$ are not required to be annotated with their respective types. 

\subsection{Reduction}
\label{subsec:reduction}

Let \contraction\ be any binary relation on terms and called a \anonymdef{contraction} relation. The \anonymdef{syntactic closure} of \contraction\ is written \stepcg\ and it is inductively defined:

\begin{definition}
\label{def:reduction}
\hfill 
\begin{code}[numbers=left]
import CFramework.CTerm as CTerm
module CFramework.CReduction (C : Set) (_▹_ : CTerm.Rel C) where
open CTerm C
...
data _⟿_ : Rel where
  abs : ∀ {x M N} → M ⟿ N → ƛ x M ⟿ ƛ x N
  appL : ∀ {M N P} → M ⟿ N → M · P ⟿ N · P
  appR : ∀ {M N P} → M ⟿ N → P · M ⟿ P · N
  redex : ∀ {M N} → M ▹ N → M ⟿ N
\end{code}
\end{definition}
\noindent Line 1 imports the module |CFramework.CTerm|, and at the same time renames it to |CTerm| just for convenience. 
Line 2 specifies that the module is parameterised by 
the contraction relation |▹|; 
notice that since we have not opened the module |CTerm| nor specified the set of constants to be used, we wrote |CTerm.Rel C| (compare with line 5).
From now until the end of this section, it is assumed that both |C| and \contraction\ are in the scope of every definition, unless explicitly stated the opposite.

Any term on the left-hand side of \contraction\ shall be called a \anonymdef{redex}, as usual, and any term on right-hand side a \anonymdef{contractum}. Besides, any term on the right-hand side of \stepcg\ shall be called a \anonymdef{reductum}.

We can define beta-reduction by means of \stepcg\ as next. Let \anonymdef{beta-contraction} be inductively defined:
\begin{definition}[Beta-contraction]
\label{def:betacontraction}
\hfill 
\begin{code}
module CFramework.CBetaContraction (C : Set) where
...
data _▹β_ : Rel where 
  beta : ∀ {x M N} → ƛ x M · N ▹β M [ N / x ]
\end{code}
\end{definition}
\noindent Then \anonymdef{beta-reduction} for the pure lambda calculus is derived by importing the modules:
\begin{definition}[Beta-reduction]
\label{def:betareduction}
\hfill 
\begin{code}
open import CFramework.CBetaContraction ⊥
open import CFramework.CReduction ⊥ _▹β_ renaming (_⟿_ to _→β_)
\end{code}
\end{definition}
\noindent Recall that in \cref{def:puresyntax} we had explained that by defining ${\tt C}=\bot$ we obtain the syntax of the pure lambda terms.

The renaming done by $\bullet$ is sensitive to the free variables in the subject term. 
As a consequence, \stepg\ is not compatible with substitution, i.e., the next lemma {\em does not} hold:
\begin{verbatim}
∀ {M N σ} → M →β N → M ∙ σ →β N ∙ σ
\end{verbatim}
Consider the following example. 
Let $M = \lambda v_1 ((\lambda v_0 \lambda v_0 v_0) v_0)$ and $N = \lambda v_1 \lambda v_0 v_0$.
It can be seen that \mbox{\step{M}{N}} is derivable. Now, let us apply $\iota$ on each side. As to $N$, 
$v_1$ is renamed to the first name fresh in the restriction \restr{\iota}{\lambda v_1 \lambda v_0 v_0}, 
i.e., to $v_0$; 
we obtain $N\bullet\iota=\lambda v_0 \lambda v_0 v_0$. 
As to $M$, the variable $v_1$ is renamed to itself, since it is the first fresh name in the corresponding restriction (renaming it to $v_0$ would cause a capture).
So, $\mbox{\effect{M}{\iota}} = M$, and the only reductum $\lambda v_1((\lambda v_0 v_0)\mbox{\unary{v_0}{v_0}})$ of $M$ equals to $\lambda v_1 \lambda v_0 v_0$, 
which is not $N\bullet\iota$.

Since we are going to need some form of the lemma of compatibility above as we shall see, we will use the next approximation which is always possible: continuing with the earlier example, after the reduction takes place we shall perform an alpha-conversion step from the reductum to meet $N\bullet\iota$, i.e., \step{M\bullet\iota}{\lambda v_1 \lambda v_0 v_0} followed by \conv{\lambda v_1 \lambda v_0 v_0}{N\bullet\iota}.

So, let $r$ be any binary relation on terms, either a contraction relation or a reduction. We shall say $r$ is \anonymdef{alpha-compatible with substitution}, and write it \texttt{Compat∙\,$r$}, if and only if, for every directed pair of terms $M$ and $N$ related under $r$, there must exist some $P$ such that  $M\bullet\sigma$ and $P$ are also related, and that \conv{P}{N\bullet\sigma}. Formally:
\begin{definition}[Alpha-compatibility with substitution]
\label{def:alphacompatsubst}
\hfill 
\begin{verbatim}
Compat∙ r = ∀ {M N σ} → r M N → Σ[ P ∈ Λ ](r (M ∙ σ) P × P ∼α N ∙ σ)
\end{verbatim}
\end{definition}
\noindent In Agda, the dependent product type can be written |Σ[ a ∈ A ] B|, where |a| is some (meta-)variable of type |A|, and |B| is some type which might depend upon |a|.

Similarly, we shall say $r$ is \anonymdef{alpha-commutative} and define it:

\begin{definition}[Alpha-commutativity]
\label{def:commutativity}
\hfill 
\begin{verbatim}
Comm∼α r = ∀ {M N P} → M ∼α N → r N P → Σ[ Q ∈ Λ ](r M Q × Q ∼α P)
\end{verbatim}
\end{definition}

We shall restrict this development to contraction relations that \anonymdef{preserve freshness}, i.e., relations that do not introduce any free name in any contractum:

\begin{definition}
|Preserves# r = ∀ {x M N} → r M N → x # M → x # N|
\end{definition}

Then we have that, if \contraction\ preserves freshness, or it is compatible with substitution, or it commutes with alpha-conversion, then its syntactic closure has the corresponding properties as well: 

\begin{lemma}
\hfill 
\label{lemma:compatSubst}
\label{lemma:commutativity}
\begin{code}
preser⟿# \,: Preserves# _▹_ → Preserves# (_⟿_ _▹_)
compat⟿∙ \,: Preserves# _▹_ → Compat∙ _▹_ → Compat∙ (_⟿_ _▹_)
commut⟿α : Preserves# _▹_ → Compat∙ _▹_ → Comm∼α _▹_ → Comm∼α (_⟿_ _▹_)
\end{code}
\end{lemma}

\noindent Their proofs are extended from \cite{tasistro2015,urciuoli2020}.
Notice the cascade effect on the lemmata: each of them has all the arguments of the one above. This happens naturally since each lemma relies on the previous one.

Finally, we have that beta-contraction is alpha-commutative, along with two other results (their proofs are extended from \cite{urciuoli2020}):

\begin{lemma}
\label{lemma:betacompatsubst}
|Preserves# _▹β_ × Compat∙ _▹β_ × Comm∼α _▹β_|
\end{lemma}

\subsection{Strongly normalizing terms}

A term is \anonymdef{strongly normalizing} if and only if every reduction path starting from it eventually halts.
We shall use their accessible characterization (originally presented in \cite{alti:phd93}). For any given computation relation \stepcg\ we define \sng:

\begin{definition}[Strongly normalizing terms]
\hfill 
\label{def:sn}
\begin{code}[numbers=left]
sn : Λ → Set 
sn = Acc (dual _⟿_) 
\end{code}
\end{definition}

\noindent |Acc| is the type of the accessible elements by some order $<$, i.e., the set of elements $a$ such that there is no infinite sequence $\ldots < a' < a$. It is defined in Agda's standard library \cite{agdastdlib}. 
|dual| is the function that returns the \textit{type} of the inverse of every binary relation on terms. We use the dual of $\rightsquigarrow$ instead of the direct because |Acc|
expects an order that descends to its left-hand side, so to speak, which is not the case for \stepcg.
Line 2 can be read as: \sng\ is the set of terms $M$ such that $M \stepcg M' \stepcg  \dots$ is always finite.
Below is the definition of |Acc| to support this paragraph:
\begin{code}
data Acc {a b} {A : Set a} (_<_ : Rel A b) (x : A) : Set (a ⊔ b) where
  acc : (∀ y → y < x → Acc _<_ y) → Acc _<_ x
\end{code}
\noindent Note that |Rel| above is the type of binary relations between any two types, and it is defined in the standard library. |a| and |b| are universe indices or levels, and |⊔| is the function that returns the greatest of them.

The next result is adapted from \cite{urciuoli2020} and follows easily by induction:
\begin{lemma}
\label{lemma:inversionSnApp}
|inversionSnApp : ∀ {M N} → sn (M · N) → sn M × sn N|
\end{lemma}

\sng\ is closed under alpha-conversion, as long as the supporting relation \stepcg\ is alpha-commutative.
The corresponding proof presented here is an adaptation of \cite{urciuoli2020}:

\begin{lemma}
\label{lemma:snClosedAlpha}
|closureSn∼α : Comm∼α _⟿_ → ∀ {M N} → sn M → M ∼α N → sn N|
\end{lemma}

\begin{proof} 
By induction on the derivation of \sn{M}. To derive \sn{N} we need to prove \sn{P} for any \stepc{N}{P}. By \cref{def:commutativity} there exists some $Q$ such that \stepc{M}{Q} and \conv{Q}{P}. 
By \cref{def:sn}, \sn{Q} holds, i.e., $Q$ is accessible, and \sn{Q} is a proper component of the derivation of \sn{P}\footnote{Put in other words, every reduction beginning in $Q$ is at least one step shorter than every other reduction beginning in $M$.}.
Then, we can use the induction hypothesis on \sn{Q} together with \conv{Q}{P} and obtain \sn{P}.
\end{proof}

\noindent Exceptionally, we show the code of the proof above because it is very compact, and to reinforce the understanding of the principle of structural induction of \sng:
\begin{code}
closureSn∼α comm {M} {N} (acc i) M∼N =
  acc λ P P←N → let Q , M→Q , Q∼P = comm M∼N P←N
                in closureSn∼α comm (i Q M→Q) Q∼P
\end{code}
The $\lambda$ occurrence denotes Agda's entity for meta-level lambda terms.
|i Q M→Q| is of type \sn{Q}, and it is a proper component of |acc i| which is of type \sn{M}. |P←N| is of type |(dual _⟿_) P N| which in turn equals to |P ⟿ N|.
For the same reason |M→Q| is of type |(dual _⟿_) Q M|.

\subsection{Reducibile terms}
\label{subsec:reducibility}

Girard's proof of the Strong Normalization Theorem defines a relation between terms and types. A term that is related to some type is said to be \anonymdef{reducible}.
The proof is carried out in two steps: first, it is proven that every reducible term is strongly normalizing, and secondly that every typed term is reducible. 
In this section we shall define the logical relation of reducible terms, and after that we shall prove some of their properties, including the first step in Girard's proof (|CR1| of \cref{lemma:redproperties}).

Both in STLC and System~T (object-level) types are simple, so regarding this development they will be enough for our definition of the logical relation. 
We define them by:
\begin{definition}[Object-level types]
\hfill
\begin{verbatim}
data Type : Set where
  τ : Type
  _⇒_ : Type → Type → Type
\end{verbatim}
\end{definition}

The \anonymdef{relation of reducible terms} or logical relation is 
then defined by recursion on the types:
\begin{definition}[Reducible terms]
\hfill 
\label{def:red}
\begin{verbatim}
Red : Type → Λ → Set
Red τ M = sn M 
Red (α ⇒ β) M = ∀ {N} → Red α N → Red β (M · N)
\end{verbatim}
\end{definition}

\Redg\ is closed under alpha-conversion:
\begin{lemma}[Closure of \Redg\ under \convg] 
\label{lemma:redClosedAlpha}
\hfill 
\begin{code}
closureRed∼α : Comm∼α _⟿_ → ∀ {α M N} → Red α M → M ∼α N → Red α N
\end{code}
\end{lemma}

\begin{proof}
By induction on the type $\alpha$, and by using \cref{lemma:snClosedAlpha}.
\end{proof}

Next we have \textit{neutral terms}. We shall use a different characterization than the one given in \cite{girard1989}, 
and define them as the set of terms 
which when applied to any non-empty sequence of arguments, the result is never a redex, i.e., $M$ is neutral if and only if $MN_1N_2\dots N_n$ is not a redex for any $n>0$. 

So first, let us define the type of \underline{vectors of applications of terms}:
\begin{definition}[Vectors]
\hfill 
\begin{code}
data Vec : Λ → Λ → Set where
  nil : ∀ {M} → Vec M M
  cons : ∀ {M N} → Vec M N → ∀ {P} → Vec M (N · P)
\end{code}
\end{definition}
\noindent \texttt{Vec M V} will then indicate that $V = MN_1N_2\dots N_n$ for some $n \geq 0$, and we shall say that $M$ is the head. If $n=0$ then $M=V$, and $M$ is not applied to any argument (we will see right away why this is convenient despite our motivation required $n>0$).

Now we can give a precise characterization of the type of \underline{neutral terms}: 
\begin{definition}[Neutral terms]
\hfill 
\begin{code}
Ne M = ∀ {V} → Vec M V → ∀ {P Q} → ¬ (V · P) ▹ Q
\end{code}
\end{definition}
\noindent 
Note that we have added $P$ at the end of |V · P| to have at least one argument applied to $M$.

The next result follows immediately by induction on the definition of |Vector|: 

\begin{lemma}
\label{lemma:appNe}
|lemmaNe : ∀ {M} → Ne M → ∀ {N} → Ne (M · N)|
\end{lemma}

As to the main result in this section we have some properties about reducible terms, among them, the first part of Girard's proof, i.e., reducible terms are strongly normalizing (|CR1|).
Let |▹| be any relation that does not reduce variables, and such that for any vector $V$ with a variable at the head it follows $V$ is neutral under 
|▹|\footnote{Actually, we could have asked the second condition just for one specific variable and \cref{lemma:redproperties} would hold anyway (see the proof of {\tt CR1} when $\alpha$ is functional).}; using our definition of vectors (possibly with no applications) we can compact both conditions by:
\begin{definition}[Condition of $\triangleright$]
\hfill 
\label{def:conditionsContraction}
\begin{code}
Cond▹ = ∀ {x N} → Vec (v x) N → ∀ {Q} → ¬ N ▹ Q
\end{code}
\end{definition}
\noindent Then, for any such a relation |▹| we have that:
\begin{lemma}[Properties of reducible terms] 
\hfill
\begin{code}
CR1 : ∀ {α M} → Red α M → sn M 
CR2 : ∀ {α M N} → Red α M → M ⟿ N → Red α N
CR3 : ∀ {α M} → Ne M → (∀ {N} → M ⟿ N → Red α N) → Red α M
\end{code}
\label{lemma:redproperties}
\end{lemma}

\begin{proof}
By mutual induction on the type $\alpha$:
\begin{itemize}
    \item Case $\alpha=\tau$:
    \begin{itemize}
        \item[{\tt CR1}] By \cref{def:red}, $\Red{\tau}{M} = \sn{M}$, so {\tt CR1} is a tautology.
        \item[{\tt CR2}] Immediate by \cref{def:sn}.
        \item[{\tt CR3}] Analogous to {\tt CR2}.
    \end{itemize}
    \item Case $\alpha=\beta\Rightarrow\gamma$:
    \begin{description}
        \item[{\tt CR1}]
        By \cref{def:conditionsContraction} we have both that \net{v_0}, and that \stepc{v_0}{N} is absurd for any $N$. As a direct consequence of the latter, the second hypothesis or argument of {\tt CR3} follows by vacuity,
        and so we can use the main induction hypothesis {\tt CR3} and obtain \Red{\beta}{v_0}. 
        Now, by \cref{def:red} on \Red{(\beta\Rightarrow\gamma)}{M}, we obtain \Red{\gamma}{(Mv_0)}, and by the induction hypothesis \sn{(Mv_0)}. 
        Finally, by \cref{lemma:inversionSnApp} we get \sn{M}.
        
        \item[{\tt CR2}] According to \cref{def:red}, to prove the thesis \Red{(\beta\Rightarrow\gamma)}{N} we have to prove \Red{\gamma}{(NP)} for any \Red{\beta}{P}. 
        By \cref{def:red}, again, the hypothesis \Red{(\beta\Rightarrow\gamma)}{M} tells us \Red{\gamma}{(MP)}, and by the \texttt{appL} rule of \cref{def:reduction} on \stepc{M}{N} we can derive \stepc{MP}{NP}, so we can use the induction hypothesis and obtain \Red{\gamma}{(NP)} as desired.
        
        \item[{\tt CR3}] Let \Red{\beta}{P}. To derive our desired result \Red{\gamma}{(MP)} and by using the induction hypothesis, we need to feed it with the required hypotheses or arguments: (1) \net{(MP)}, and (2) that for every $N'$, \stepc{MP}{N'} implies \Red{\gamma}{N'}. (1) follows by \cref{lemma:appNe}.
        As to (2), first of all, by the main induction hypothesis {\tt CR1} we get \sn{P}. Now we continue by a nested induction on the derivation of \sn{P}\footnote{In the code, it means to define an auxiliary function in the current scope.}. Let us analyse every possible derivation of \stepc{MP}{N'}.
        \begin{itemize}
            \item Case {\tt redex}: \contracts{MP}{N'} is absurd by \cref{def:conditionsContraction}.
            
            \item Case {\tt appL}: If \stepc{MP}{M'N''} follows from \stepc{M}{M'} with $N'=M'N''$ then by (2) we get \mbox{\Red{(\beta\Rightarrow\gamma)}{M'}}, and so by \cref{def:red}, \Red{\gamma}{(M'N'')}. 
            
            \item Case {\tt appR}: If \stepc{MP}{MP'} follows from \stepc{P}{P'} with $N'=MP'$, then by \cref{def:sn} we obtain \sn{P'}, which is a proper component of \sn{P}, and so we can continue by induction on \sn{P'}.
        \end{itemize}
    \end{description}
\end{itemize}
\end{proof}

Next we have some general definitions regarding the assignment of types.
First, there are \anonymdef{contexts} (of variable declarations). They are defined as list of pairs, possibly with duplicates:

\begin{definition}
|Cxt = List (V × Type)|
\end{definition}

\noindent Then there is the relation of \anonymdef{membership} between variables and contexts.
We shall write $x\in\Gamma$ and say that $x$ is the \emph{first} variable in $\Gamma$, searched from left to right. Below is the inductive definition:

\begin{code}
data _∈_ : V → Cxt → Set where
  here  : ∀ {x α Γ} → x ∈ Γ $\Cup$ x ∶ α
  there : ∀ {x y α Γ} → x ≢ y → x ∈ Γ → x ∈ Γ $\Cup$ y ∶ α 
\end{code}

\noindent |$\Gamma$ $\Cup$ x : $\alpha$| is syntax-sugar for |(x , $\alpha$) ∷ $\Gamma$|. 
Finally, there is a \anonymdef{lookup} function on contexts such that it returns the type of the first variable (provided it is declared), searched in the same fashion, and defined:
\begin{code}[numbers=left]
_⟨_⟩ : ∀ {x} → (Γ : Cxt) → x ∈ Γ → Type
[]              ⟨ ()        ⟩ 
((k , d) ∷ xs) ⟨ here      ⟩ = d
((k , d) ∷ xs) ⟨ there _ p ⟩ = xs ⟨ p ⟩
\end{code}
\noindent In the second line, {\tt ()} is an \emph{absurd} pattern, and it tells Agda to check that there is no possible way of having an object of type $x\in\texttt{[]}$, for any $x$.

To end this section, we present \anonymdef{reducible substitutions}. We shall say a substitution is reducible under some context $\Gamma$ if and only if it maps every variable in $\Gamma$ to a reducible term of the same type:

\begin{definition}
\label{def:redsubst}
|RedSubst σ Γ = ∀ x → (k : x ∈ Γ) → Red (Γ ⟨ k ⟩) (σ x)|
\end{definition}

\noindent The next results follow immediately by definition:

\begin{lemma}
\label{lemma:iotaisred}
\begin{verbatim}
Red-ι : ∀ {Γ} → RedSubst ι Γ
\end{verbatim}
\end{lemma}

\begin{lemma}
\label{lemma:updateRedSubst}
\hfill 
\begin{code} 
Red-upd : RedSubst σ Γ → ∀ x → Red α N → RedSubst (σ ≺+ (x , N)) (Γ $\Cup$ x : α)
\end{code}
\end{lemma}

\section{STLC}
\label{sec:stlc}

The syntax and theories of substitution, alpha- and beta-conversion for STLC are obtained by instantiating the framework with:
\begin{code}
module STLC where
open import CFramework.CTerm ⊥  
...
open import CFramework.CReduction ⊥ _▹β_ as Reduction renaming (_⟿_ to _→β_)
\end{code}

Next is the assignment of types in STLC:

\begin{code}
data _⊢_∶_ (Γ : Cxt) : $\Lambda$ → Type → Set where
  ⊢var : ∀ {x} → (k : x ∈ Γ) → Γ ⊢ v x ∶ Γ ⟨ k ⟩
  ⊢abs : ∀ {x M α β} → Γ $\Cup$ x ∶ α ⊢ M ∶ β → Γ ⊢ ƛ x M ∶ α ⇒ β
  ⊢app : ∀ {M N α β} → Γ ⊢ M ∶ α ⇒ β → Γ ⊢ N ∶ α → Γ ⊢ M · N ∶ β
\end{code}

\subsection{The Strong Normalization Theorem in STLC}
\label{sec:normalizationSTLC}

Following Girard's proof, first we need to prove that every reducible term is \sng. We shall use |CR1| of \cref{lemma:redproperties} for that matter, so we need to prove that \betacontraction\ satisfies conditions in \cref{def:conditionsContraction}.

\begin{lemma}
|cond▹β : ∀ {x N} → Vec (v x) N → ∀ {Q} → ¬(N ▹β Q)|
\end{lemma}
\begin{proof}
Immediate by contradiction from |Vec (x N) N| and |N ▹β Q|.
\end{proof}
\noindent Then we can open the following modules and inherit \cref{lemma:redproperties} for STLC, particularly |CR1|:
\begin{code}
open import CFramework.CReducibility ⊥ _▹β_ as Reducibility
open Reducibility.RedProperties cond▹β
\end{code}

Next we have to prove that every typed terms is reducible; we shall refer to this as the \anonymdef{main} lemma. To present the proof, we are going to need some preparatory results.
First, by \cref{lemma:betacompatsubst} together with \cref{lemma:compatSubst} we have that \stepg\ is both alpha-compatible with substitution, and alpha-commutative:

\begin{lemma}
\label{lemma:stepBetacompatSubst}

\label{lemma:stepBetaCommutesAlpha}
|Compat∙ _→β_ × Comm∼α _→β_|
\end{lemma}

\noindent Secondly, it is immediate that $(\lambda x N)N$ is neutral for every $x$, $M$ and $N$:

\begin{lemma}
|lemmaβNe : ∀ {x M N} → Ne ((ƛ x M) · N)|
\end{lemma}

\noindent And finally, since the main lemma proceeds by induction on the derivation of the typing judgement, and the case of abstractions is quite complex, it turns out to be convenient to have a separate lemma for this:

\begin{lemma}
\label{lemma:abstraction}
|lemmaAbs : ∀ {x M N α β} → sn M → sn N| 
\hfill
\begin{code}
→ (∀ {P} → Red α P → Red β (M [ P / x ])) → Red α N → Red β (ƛ x M · N)
\end{code}
\end{lemma}

\begin{proof}
By induction on the derivations of \sn{M} and \sn{N}. 
We shall refer to hypotheses \sn{M}, \sn{N}, $\forall \{P\}\ \mbox{\tt →}\ \Red{\alpha}{P}\ \mbox{\tt →}\ \Red{\beta}{(M\mbox{\unary{P}{x}})}$ and \Red{\alpha}{N} as (1) through (4) respectively.
So, to use |CR3| of \cref{lemma:redproperties} to prove that the 
neutral term
$(\lambda x M)N$ is reducible of type $\beta$ (the thesis of this lemma) we need to show that every reductum  
is reducible (the second explicit hypothesis of the mentioned lemma). So, let us analyze every possible case:
\begin{itemize}
    \item Case {\tt redex}: If \step{(\lambda x M)N}{M\mbox{\unary{N}{x}}} then we can quickly derive that $M\mbox{\unary{N}{x}}$ is reducible from (3) and (4). 
    
    \item Case {\tt appL}: If \mbox{\step{(\lambda x M)N}{(\lambda x M')N}} follows from \mbox{\step{M}{M'}} then, to use the induction hypothesis on \sn{M'}, we need to provide the requested hypotheses (1) through (4) correctly instantiated. 
    (1) follows from \cref{def:sn}. (2) and (4) are direct. 
    As to (3), we need to prove that \Red{\beta}{(M'\mbox{\unary{P}{x}})} holds for any \Red{\alpha}{P}. By \cref{lemma:stepBetacompatSubst} we know that there exists some $R$ such that \step{M\mbox{\unary{P}{x}}}{R} and \conv{R}{M'\mbox{\unary{P}{x}}}. By hypothesis (3) it follows \Red{\beta}{(M\mbox{\unary{P}{x}})}, so by |CR2| of \cref{lemma:redproperties} we obtain \Red{\beta}{R}. Finally, we can use \cref{lemma:stepBetaCommutesAlpha} together with inherited \cref{lemma:redClosedAlpha} to derive \Red{\beta}{(M'\mbox{\unary{P}{x}})}. 
    
    \item Case {\tt appR}: If \mbox{\step{(\lambda x M)N}{(\lambda x M)N'}} follows from \step{N}{N'} then, by \cref{def:sn} we have \sn{N'}, and by |CR2| of \cref{lemma:redproperties} we obtain \Red{\alpha}{N'}, therefore we can use the induction hypothesis on \sn{N'} and derive \Red{\beta}{((\lambda x M)N')}.
\end{itemize}
\end{proof}

Now, to use the previous result in the main lemma, we are going to need a stronger induction hypothesis in order to derive the third argument, namely
$\forall \{P\}\ \mbox{\tt →}\ \Red{\alpha}{P}\ \mbox{\tt →}\ \Red{\beta}{(M\mbox{\unary{P}{x}})}$.
We shall see that by stating the main lemma as next we can easily derive it:

\begin{lemma}
\label{lemma:subst}
|main : ∀ {α M σ Γ} → Γ ⊢ M ∶ α → RedSubst σ Γ → Red α (M ∙ σ)|
\end{lemma}

\begin{proof}
By induction on the typing derivation:
\begin{itemize}
    \item Case {\tt $\vdash$var}: If $M$ is a variable, then the thesis follows directly from \cref{def:redsubst}.
    
    \item Case {\tt $\vdash$abs}: If $M=\lambda x M'$ with type $\alpha\Rightarrow\beta$, then we need to show \Red{\beta}{((\mbox{\effect{(\lambda x M')}{\sigma}})N)} for any \Red{\alpha}{N}. 
    First of all, $\mbox{\effect{\lambda x M'}{\sigma}}=\lambda y (\mbox{\effect{M'}{(\mbox{\upd{\sigma}{x}{y}})}})$ for some fresh name $y$.
    Now, to use \cref{lemma:abstraction} we need to derive its hypothesis: (1) \sn{(\mbox{\effect{M'}{\mbox{\upd{\sigma}{x}{y}}}})}; (2) \sn{N}; (3) for every \Red{\alpha}{P}, \Red{\beta}{((\mbox{\effect{M'}{\mbox{\upd{\sigma}{x}{y}}}})\mbox{\unary{P}{y}})}, and; (4) \Red{\alpha}{N}.
    As to (1), by \cref{lemma:updateRedSubst} we have \RedSubst{(\Gamma\Cup x:\alpha)}{(\mbox{\upd{\sigma}{x}{y}})}, thus by induction hypothesis \Red{\beta}{(M'\bullet\mbox{\upd{\sigma}{x}{y}})}, and so by |CR1| of \cref{lemma:redproperties} we obtain the desired result. (2) follows immediately by |CR1|. As to (3), first by \cref{lemma:composition} we have \conv{(\mbox{\effect{M'}{\mbox{\upd{\sigma}{x}{y}}}})\mbox{\unary{P}{y}}}{M'\bullet\mbox{\upd{\sigma}{x}{P}}}. 
    Next, by \cref{lemma:updateRedSubst}, \RedSubst{(\Gamma\Cup x:\alpha)}{(\mbox{\upd{\sigma}{x}{P}})}, so by the induction hypothesis we have \Red{\beta}{(M'(\sigma,P/x))}. And finally, by \cref{lemma:stepBetaCommutesAlpha} together with \cref{lemma:redClosedAlpha} we can derive the desired result. 
    (4) is an assumption already made. At last, having (1) through (4) we can use \cref{lemma:abstraction} and derive \Red{(\alpha\Rightarrow\beta)}{(\mbox{\effect{(\lambda x M')}{\sigma}})}, and so obtain \Red{\beta}{((\mbox{\effect{(\lambda x M')}{\sigma}})N)} by \cref{def:red}, as desired.
    
    \item Case {\tt $\vdash$app}: Immediate by the induction hypothesis.
\end{itemize}
\end{proof}

Without further ado, we have the Strong Normalization Theorem:

\begin{theorem}
\label{theo:sn}
|strongNormalization : ∀ {Γ M α} → Γ ⊢ M ∶ α → sn M|
\end{theorem}

\begin{proof}
By \cref{lemma:iotaisred,lemma:subst} we have \Red{\alpha}{(\mbox{\effect{M}{\iota}})}, and so by |CR1| of \cref{lemma:redproperties}, \sn{(\mbox{\effect{M}{\iota}})}. Then, by \cref{lemma:substIota}, \conv{\mbox{\effect{M}{\iota}}}{M}, and thus by \cref{lemma:stepBetaCommutesAlpha} together with \cref{lemma:snClosedAlpha} it follows \sn{M}.
\end{proof}

\section{System T}
\label{sec:systemt}

Let |C| and |\contraction\!T| be inductively defined:
\begin{code}
data C : Set where 
  O : C; S : C; Rec : C
\end{code}

\begin{code}
data _▹T_ : Rel where
  beta : ∀ {M N} → M ▹β N → M ▹T N
  recO : ∀ {G H} → k Rec · G · H · k O ▹T G
  recS : ∀ {G H N} → k Rec · G · H · (k S · N) ▹T H · N · (k Rec · G · H · N)
\end{code}
The syntax and theories of substitution, alpha- and beta-conversion for System~T are then obtained by instantiating the framework with both |C| and |\contraction\!T|, and similarly to STLC as shown in the previous section.

The assignment of types in System~T is extended from STLC and defined:
\begin{code}
data _⊢_∶_ (Γ : Cxt) : Λ → Type → Set where
  ⊢zro : Γ ⊢ k O ∶ nat
  ⊢suc : Γ ⊢ k S ∶ nat ⇒ nat
  ⊢rec : ∀ {α} → Γ ⊢ k Rec ∶ α ⇒ (nat ⇒ α ⇒ α) ⇒ nat ⇒ α
  ⊢var : ∀ {x} → (k : x ∈ Γ) → Γ ⊢ v x ∶ Γ ⟨ k ⟩
  ⊢abs : ∀ {x M α β} → Γ $\Cup$ x ∶ α ⊢ M ∶ β → Γ ⊢ ƛ x M ∶ α ⇒ β
  ⊢app : ∀ {M N α β} → Γ ⊢ M ∶ α ⇒ β → Γ ⊢ N ∶ α → Γ ⊢ M · N ∶ β
\end{code}
\noindent |nat| is syntax-sugar for $\tau$.

\subsection{The Strong Normalization Theorem in System~T}

The proof of the Strong Normalization Theorem for System~T follows the same structure as the one for STLC: first, we have to prove that |▹T| satisfies condition in \cref{def:conditionsContraction} so to derive the first step in Girard's method, i.e., |CR1|. Then, we need to have a main lemma and reason by induction on the syntax (the typing judgment) to derive reducibiliy. 
Finally, the Strong Normalization Theorem for System~T follows {\em exactly} as \cref{theo:sn}.

So, to start with, we have that |▹T| satisfies \cref{def:conditionsContraction} similar to STLC:
\begin{lemma}
|cond▹T : ∀ {x N} → Vec (v x) N → ∀ {Q} → ¬(N ▹T Q)|
\end{lemma}
\noindent
Thus, we inherit \cref{lemma:redproperties} in System~T, particularly |CR1|.

As to the second part, i.e., the main lemma, we have to consider only the additional syntax;
the remaining cases follow identically.
\zero\ and \sucg\ are reducible by |CR3| (similar to $v_0$ in the proof of {\tt CR1}). As to \recg, we shall follow the same strategy as in STLC and have a separate lemma, namely the \anonymdef{recursion lemma}. In the next section we cover this last case, while at the same time we present the announced simplification.

\subsection{Recursion}
\label{subse:rec}

In this section, first we explain the induction used in the proof of the recursion lemma as presented in \cite{girard1989} but using the terminology of our framework, then we present the simplification, and finally we formalize the proof.

We must prove that the neutral term \rec{G}{H}{N} is reducible for any reducible terms $G$, $H$ and $N$. 
First, we shall strengthen our induction hypothesis: by |CR1| we know that $G$, $H$ and $N$ are \sng, so we can assume that these derivations are given as additional hypotheses. 
Also, we need some preparatory definitions: let $\nu(M)$, $\ell(M)$ and $\mathsf{nf}(M)$ be respectively the length of the longest reduction starting in $M$, the count of \sucg\ symbols in $M$, and the normal form of the (strongly normalizing) term $M$.
Now, to prove our thesis we shall proceed by induction on the {\em strict component-wise} order (henceforth, just component-wise order) on the 4-tuple\footnote{The component-wise order on a $n$-tuple is given by: $a_i <_i b \Rightarrow (a_0\ldots,a_i\ldots,a_n) <_i (a_0\ldots,b\ldots,a_n)$ for any $i$, $n$ and $b$.} 
$(\sn{G}, \sn{H}, \nu(N), \ell(\mathsf{nf}(N)))$,
where in \sn{G} and \sn{H} we shall use the structural order of \sng,
in $\nu(N)$
the complete order on natural numbers\footnote{The complete order on natural numbers is the same as transitive closure of the structural order on them.}, 
and in $\ell(\mathsf{nf}(N))$ the structural order on natural numbers.
As we did in \cref{lemma:abstraction}, we are going to use |CR3| of \cref{lemma:redproperties} for the matter, and so
we have to prove that every reductum of \rec{G}{H}{N} is reducible.
There are five cases: 
\begin{enumerate*}[label=(\arabic*)]
    \item \rec{G'}{H}{N} with \step{G}{G'}, 
    \item \rec{G}{H'}{N} with \step{H}{H'}, 
    \item \rec{G}{H}{N'} with \mbox{\step{N}{N'}}, 
    \item $G$ with $N=\zero$, and 
    \item $HN'(\rec{G}{H}{N'})$ with $N=\suc{N'}$. 
\end{enumerate*}
As to (1) and (2), we can directly use the induction hypothesis on $\sn{G'}$ and  $\sn{H'}$.
As to (3), we can suspect that $\nu(N')<\nu(N)$, and so we can proceed likewise. (4) is a hypothesis. As to (5), it is immediate that $\ell(\mathsf{nf}(N'))<\ell(\mathsf{nf}(\suc{N'}))$.

We can simplify the induction schema used above by dispensing with $\mathsf{nf}$, and instead proceed by induction on the component-wise order of the 3-tuple 
$(\sn{G}, \sn{H}, (\nu(N), \ell(N)))$,
where in \sn{G} and \sn{H} we use the same order as above, but in $(\nu(N), \ell(N))$ we use the {\em lexicographic order} on tuples\footnote{The lexicographic order on a tuple is given by: $a<b \Rightarrow (a,c)<(b,d)$ and $b<c \Rightarrow (a,b)<(a,c)$ for any $a$, $b$, $c$, $d$.}.
As to cases (1), (2) and (4), the induction is the same. 
As to (3), we have already assumed that \mbox{$\nu(N')<\nu(N)$}, so we can use the (lexicographic-based) induction hypothesis on $(\nu(N'),\ell(N'))$, and disregard if $\ell(N')$ goes off.
Finally, as to (5), on the one hand, it is immediate that $\ell(N')<\ell(\suc{N'})$. On the other hand, we can also guess that $\nu(N')=\nu(\suc{N'})$, therefore we can proceed by induction on $(\nu(\suc{N'}),\ell(N'))$. 

Now, to formalize the recursion lemma based on the last induction schema, first we need to give some definitions, as usual. Next is the function that computes the list of reductio for any given term $M$, while at the same time proves it is {\em sound}, i.e., every element of the list is actually a reductum of $M$.
We present it in two separate parts, first |redAux|, which as the name suggest, is an auxiliary function, and then |reductio| which is the complete and desired operation (we omit some code):

\begin{code}[numbers=left]
redAux : (M : Term) → List (Σ[ N ∈ Term ](M →β N))
redAux (ƛ x M · N)                 = [ (M [ N / x ] , ...) ]
redAux (k Rec · G · H · k O)       = [ (G , ...) ]
redAux (k Rec · G · H · (k S · N)) = [ (H · N · (k Rec · G · H · N) , ...) ]
redAux _                           = []

reductio : (M : Term) → List (Σ[ N ∈ Term ](M →β N))
reductio (k _)   = []
reductio (v _)   = []
reductio (ƛ x M) = mapL (mapΣ (ƛ x) abs) (reductio M)
reductio (M · N) = redAux (M · N) ++ ... (reductio M) ++ ... (reductio N)
\end{code}

\noindent |mapΣ| is the function that given two other functions and a tuple, it applies each function to one of the components of the tuple. The purpose of the auxiliary function is to put together the cases of redexes, and apart from the |reductio| definition, 
so to have a cleaner treatment in the case of applications in the latter (see line 11).

The algorithm is also {\em complete}, i.e., it outputs all reductio of $M$, and its proof follows by induction on the derivation of any given reduction:
\begin{lemma}
\label{lemma:reduction}
|lemmaReductio : ∀ {M N} (r : M →β N) → (N , r) ∈′ (reductio M)|
\end{lemma}

\noindent |∈′| is the standard relation of membership in lists.

We can use the list returned by |reductio| to develop an algorithm that computes our first ordinal $\nu$, i.e., the length of the longest reduction beginning in some strongly normalizing term $M$ given, by recursively computing such a result for every reductum of $M$, then selecting the longest one, and finally adding one for the first step.
Notice that the length of longest path and the height of the derivation tree of \sn{M} are synonyms, so we shall use them interchangeably:
\begin{code}
ν : ∀ {M} → sn M → ℕ 
ν {M} (acc i) = 1 + max (mapL (λ{(N , M→N) → ν (i N M→N)}) (reductio M))
\end{code}
|max| is the function that returns the maximum element in a given list.
The above definition is standard for computing the height of any inductive type, except for that \sng\ has an {\em infinitary premise} \cite[p.~13]{abelTwelf}. This means that we need to enumerate all possible applications to obtain every possible sub-tree. Since every term $M$ has a finite number of redexes, so there can only be finitely many applications of the premise, i.e., reductions \step{M}{N} for some $N$, all of them being enumerated by the {\tt reductio} algorithm, as proven in \cref{lemma:reduction}.

The height of \sn{N} equals to the height of \sn{(\suc{N})}, as guessed at the start of this section. This is immediate since the prefix \sucg\ does not add any redex to any reduction path:

\begin{lemma}
\label{lemma:heightSuc}
|lemmaSν : ∀ {M} (p : sn M) (q : sn (k S · M)) → ν p ≡ ν q|
\end{lemma}

\begin{proof}
By induction on either the derivation of $p$ or $q$.
\end{proof}

Next we have that the height of \sn{M} decreases after a computation step is consumed, or in other words, every (immediate) sub-tree of \sn{M} is strictly smaller. 
The name of the lemma is |lemmaStepν|, and 
its proof follows by properties of lists, and by using \cref{lemma:reduction}:

\begin{lemma}
\label{lemma:heightStep}
|∀ {M N i} (p : sn M) → p ≡ acc i → (r : M →β N) → ν (i N r) < ν p|
\end{lemma}

\noindent Notice the {\em apparently} clumsy way it was stated. |i N r| is a proof of \sn{N}. To require such a proof as an argument would be inefficient since we already know \sn{M} and \step{M}{N}. 
Instead, by asking for the argument |p ≡ acc i| we can obtain the premise |i| of \sn{M} (this can be easily supplied afterwards with the constructor of |≡|, |refl|), and apply it to |N| and |r|, and so obtain a proof of \sn{N}.

Next is our second ordinal:

\begin{definition}
\label{def:symbols}
|ℓ : Term → ℕ| {\em is the function that counts the number of occurrences of the \sucg\ symbol in any given term, and it is defined by recursion on the term.}
\end{definition}

Finally, we have the recursion lemma.
Let |<-lex| be the lexicographic order on tuples of $\mathbb{N}$. Then |Acc _<-lex_| is the type of pairs that are accessible by such an order;
it is easy to prove that for any proof $p$ of \sn{N} for some $N$, it follows that $(\nu(p) , \ell(N))$ is in the accessible part of the lexicographic order, hence such an argument can always be derived.
Also, note that \rec{G}{H}{N} is neutral for any $G$, $H$ and $N$.
Then:

\begin{lemma}[Recursion]
\hfill 
\begin{code}
lemmaRec : ∀ {α G H N} → sn G → sn H → (p : sn N) → Acc _<-lex_ (ν p , ℓ N)
→ Red α G → Red (nat ⇒ α ⇒ α) H → Red α (k Rec · G · H · N)
\end{code}
\end{lemma}

\begin{proof}
By induction on the derivations of \sn{G} and \sn{H}, and on the lexicographic order of the tuple $(\nu(p),\ell(N))$\footnote{In Agda, every function is structural recursive, and each one of them will successfully pass the type-checking phase if, put it simply, there exists a subset of the arguments such that for every recursive call in any of its definiens, at least one of the arguments is structurally smaller whilst the others remains the same. This is the same as saying that the induction is based on the component-wise order of any arrangement of such a subset, i.e., on a tuple made up of such arguments.}. 
As has already been said several times by now, we shall resort to |CR3| of \cref{lemma:redproperties} for the matter. So let us fast-forward til the reductum analysis: 
\begin{itemize}
    \item Case {\tt recO}: If \step{\rec{G}{H}{\zero}}{G} then the result is a hypothesis.
    
    \item Case {\tt recS}: If \step{\rec{G}{H}{(\suc{N})}}{HN(\rec{G}{H}{N})} then, since we know both that \mbox{$\nu(N)=\nu(\suc{N})$} by \cref{lemma:heightSuc}, and that $\ell(N) < \ell(\suc{N})$ by definition of $\ell$, we can apply the induction hypothesis and so obtain \Red{\alpha}{(\rec{G}{H}{N})}. Finally, by \cref{def:red} on \Red{({\tt nat}\Rightarrow \alpha \Rightarrow \alpha)}{H} we obtain \Red{\alpha}{(HN(\rec{G}{H}{N}))}.
    
    \item Case {\tt appR}: If \step{\rec{G}{H}{N}}{\rec{G}{H}{N'}} follows from \step{N}{N'} then, by \cref{lemma:heightStep} we know that \mbox{$\nu(N')<\nu(N)$}, and so we can use the induction hypothesis to derive the desired result. 
    
    \item Case {\tt appL}: If the reduction follows from one either in $G$ or $H$, then we can proceed directly by the induction hypothesis.
\end{itemize}
\end{proof}

\section{Related work}

In this development we have encoded the lambda terms using first-order abstract syntax (FOAS).
In contrast, other approaches use \textit{higher-order abstract syntax} (HOAS) \cite{hoas}, i.e., binders and variables are encoded
using the same ones in the host language. 
These systems have the advantage that substitution is already defined.
The first such mechanization of the theorem for STLC was presented in \cite{donnelly2007}, and by using the ATS/LF logical framework \cite{ats}.
However, the theory of (terminating) recursive functions using FOAS is more established, and there are plenty of programming languages that support them.  
This makes fairly easy to translate this mechanization to other 
system supporting standard principles of induction.

A second difference with existing work is that in this paper we have used named variables instead of \textit{de-Bruijn indices} \cite{debruijn1972}, e.g., in our framework the identity function can be written $\lambda x x$ for any $x$, while in the latter $\lambda 0$. Clearly, the former is visually more appealing, making it better suited for textbooks, needless to say it is the actual way programs are written. The main disadvantage is that we do not identify alpha-convertible terms, e.g., $\lambda v_0 v_0$ and $\lambda v_1 v_1$ are different objects, whereas by using indices there is only one possible representative for each class of alpha-convertible terms, and so it is not necessary to deal with alpha-conversion at all. 
To mention some renowned mechanizations of the theorem for STLC using this encoding: in \cite{alti:phd93} the author uses the LEGO system \cite{lego}, and;
in \cite{poplmark2} two different mechanizations are presented, one in Agda and one in Coq \cite{coq}.

As to System T, to the best of our knowledge there is not yet a mechanization of the Strong Normalization Theorem. 


\section{Conclusions}
\label{sec:conclusions}

We have presented a framework for the formal meta-theory of lambda calculi in FOAS with constants, that does not identify alpha-convertible terms, and it is parameterised by a reduction schema. 
On top of it, we have built a complete mechanization of Girard's proof of the Strong Normalization Theorem for System~T.
In addition, we were able to include a simplification on the principle of induction of the original proof. 
Finally, we gave a new and different mechanization of the same method but for STLC.

In terms of size, the framework is ${\sim}1800$LOC long, counting import statements and the like, and of which ${\sim}90$LOC belong to the first part of Girard's proof, namely the reducibility properties.
As to the mechanizations of the proofs for STLC and System~T, they are about $70$ and $260$LOC long repectively.

The proof for STLC presented here is significantly shorter than that of previous works using the same framework. 
In \cite{urciuoli2020}, a proof of the Strong Normalization Theorem for STLC using Joachimski and Matthes' method was presented, and soon after, it was refactored to take alpha-conversion out of the syntactic characterization of the strongly normalizing terms.
The final proof was ${\sim}400$LOC long.
The mechanization presented here adds up to ${\sim}160$LOC,
i.e., less than half the size.
One of the main differences is that the closure of the accessibility definition of the strongly normalizing terms under alpha-conversion required just 3LOC, while its syntactical counterpart required about $100$LOC.

Overall, during this work alpha-conversion was not much of a burden outside the framework. Once the machinery has been set up, just a handful of lemmas were used at specific locations. Beta-reduction was proven to be both alpha-commutative and alpha-compatible with substitution in \cref{lemma:stepBetacompatSubst}, and after that, both results were used in \cref{lemma:abstraction,lemma:subst,theo:sn}, along with \cref{lemma:substIota,lemma:composition,lemma:snClosedAlpha,lemma:redClosedAlpha}, all of them having been previously defined in the framework. Alpha-conversion was not used at all in the recursion lemma.

We hope that this paper can also serve as a tool to extend the proof method to related calculi and different host languages. The method we have presented uses simple techniques and it is rich in details, 
so hopefully it can be adjusted to different scenarios.

\nocite{*}
\bibliographystyle{eptcs}
\bibliography{references}

\end{document}